\documentclass[conference,a4paper]{IEEEtran}

\newcommand{\textred}[1]{\textcolor{red}{#1}}

\ifx\noeditingmarks\undefined
  \newcommand{\pgwrapper}[2]{\textred{#1: #2}}
\else
  \newcommand{\pgwrapper}[2]{}
\fi
\usepackage{soul}
\usepackage{framed}
\usepackage{bbm}
\usepackage{epsfig}
\usepackage{times}
\usepackage{float}
\usepackage{afterpage}
\usepackage{amsmath}
\usepackage{amstext}
\usepackage{amssymb,bm}
\usepackage{latexsym}
\usepackage{color}
\usepackage{graphicx}
\usepackage{amsmath}
\usepackage{amsthm}
\usepackage{graphicx}
\usepackage[center]{caption}
\usepackage{caption}
\usepackage{subcaption}
\usepackage{booktabs}
\usepackage{multicol}
\usepackage{lipsum}% dummy text
\usepackage{verbatim}

\newtheorem{lemma}{Lemma}

\newtheorem{theorem}{Theorem}

\usepackage[usenames,dvipsnames]{xcolor}

\newcommand{\widesim}[2][1.5]{
  \mathrel{\overset{#2}{\scalebox{#1}[1]{$\sim$}}}
}
\newcommand{\PP}{\mathbb{P}}
\newcommand{\simiid}{\widesim{\text{i.i.d}}}

\newcommand{\hsig}{\hat{\sigma}}
\newcommand{\Hsig}{\hat{\Sigma}}

\begin{document}

\title{On Identifying a Massive Number of Distributions}
\author{
\IEEEauthorblockN{
Sara Shahi, Daniela Tuninetti and Natasha Devroye\\}
%\medskip
\IEEEauthorblockA{%
%Department of Electrical and Computer Engineering,
University of Illinois at Chicago, Chicago IL 60607, USA. \\
Email: {\tt sshahi7, danielat, devroye @uic.edu}}%
}

\maketitle

\begin{abstract}
THIS PAPER IS ELIGIBLE FOR THE STUDENT PAPER AWARD.
Finding the underlying probability distributions of a set of observed sequences under the constraint that each sequence is generated i.i.d by a distinct distribution is considered. 
The number of distributions, and hence the number of observed sequences, are let to grow with the observation blocklength $n$. Asymptotically matching upper and lower bounds on the probability of error are derived. 
\end{abstract}

\section{Introduction}
\label{sec:intro}
Hypothesis testing is a classical problem in statistics where one is given a random observation vector and one seeks to identify the distribution from a given set of distributions that generated it. 
% given a set of distinct distributions and an observation vector generated by one of these distributions, one is interested iitunes finding its true generating distribution. 
%Hypothesis testing is also closely tied to many
%%It also has close ties to many 
%problems in information theory. 
%Numerous works has been done on the asymptotic  behavior of different error types in this problem. Some of the p
Pioneering work in classical hypothesis testing include the proof of the 
optimality of likelihood ratio tests under certain criteria in the  Neyman-Pearon Theorem~\cite{neyman1933problem}. 
Derivation of error exponents of different error types and their trade-offs for binary and M-ary hypothesis testing in~\cite{blahut1974hypothesis} and~\cite{tuncel2005extensions} and the analysis of 
sequential hypothesis testing in~\cite{wald1945sequential}.% in which it was shown that a sequential  test can achieve strictly higher error exponents than fixed length tests.
%Bechhofer, Kiefer, Sobel identification problem?\\

The classical identification problem, which includes hypothesis testing as a special case,  is consist of a finite number of distinct sources, each generating a sequence of i.i.d samples. The problem is to find the  underlying distribution of each sample sequence, given the constraint that each sequence is generated by a distinct distribution. 
With this constraint the number of hypothesis is exponential in the number of distributions. 
If one neglects the fact that the sequences are generated by distinct distributions, %the number of hypothesis reduces  to the square of the number of distributions and boils
 the problem boils down to multiple M-ary hypothesis testing problems. 
This approach is suboptimal as it fails to exploit some of the (possibly useful) constraints. 

In~\cite{ahlswede2006logarithmically}, the authors study the the Logarithmically Asymptotically Optimal (LAO) Testing of identification problem for a  finite number of distributions. In particular, they study the identification of only two different objects in detail  and find the reliability matrix, which consist of the error exponents of all error types. Their optimality criterion is to find the largest error exponent  for a set of error types for given values of the other error types error exponent. The same problem with a different optimality criterion was also studied  in~\cite{unnikrishnan2015asymptotically}, where multiple, finite, sequences were matched to the source distributions. More specifically, they proposed a test for a generalized Neyman-Pearson-like optimality criterion to minimize the rejection probability given that all other error probabilities decay exponentially with a pre-specified slope. 

In here, we assume $A$ sequences of length $n$ are generated i.i.d according to $A$ distinct distributions; in particular random vectors $X_i^n  \simiid P_{\sigma_i},i\in[1:A]$, for some unknown permutation $\sigma$ of the distributions. The goal is to reliably identify the permutation $\sigma$ with vanishing error probability as $n\to\infty$ from an observation of $[X_1^n,\ldots, X_A^n]$. This problem has close ties with  de-anonymization of anonymized data~\cite{unnikrishnan2015asymptotically}. A different motivation is the identification of users using only channel output sequences, without the use of pilot / explicit identification signals~\cite{shahi2016isit}.  In both scenarios, the problem's difficulty increases with the number of users. In addition, in modeling the systems with a massive number of users (such as the Internet of Things), it may be reasonable to assume that the number of users grow with the transmission blocklength~\cite{shahi2016isit},~\cite{chen2017capacity}, and that the  user's identities must be distinguished from the received data.  
 As the result, it is useful to understand exactly how the number of distributions affects the system performance, in particular for the case that the cardinality of the distributions grows  with the blocklength. Notice that in this scenario, the number of hypothesis, would be doubly exponential in blocklength and the analysis of the optimal decoder becomes much harder than the classical (with constant number of distributions) identification problems. 

{\bf Contributions.} In this paper, we consider the identification problem  for the case that the number of distributions grow with the observation blocklength $n$ as motivated by the massive user identification problem in the Internet of Things  paradigm. The key novel element in this work consist of analyzing and reducing the complexity of the optimal maximum likelihood decoder, with double exponential number of hypothesis, using a graph theoretic result. In particular, we find 
\begin{enumerate}
\item Find matching upper and lower bounds on the probability of  error. This result specifies the relation between  the growth rate of the number of distributions and the pairwise distance of the distributions for reliable identification.
\item We show that the probability that more than two distributions are incorrectly identified  is dominated by the probability of the event that only two distributions are incorrectly identified.
\item  
We show that the arithmetic mean of the cycles gains (where we define the cycle gain as the product of the edge weights within the cycle) in a graph can be upper bounded by a function of  the sum of the squares of the edge weights. This may be of independent interest.
\end{enumerate}
%We end this section by introducing a set of special notations. In section~\ref{sec:problem}, we introduce the problem formulation and the main result of this paper. The main theorem proof consist  of achievability and converse bounds

\section{Notation}
 Capital letters represent random variables that take on lower case letter values in calligraphic letter alphabets. For a set of finite alphabet $\mathcal{X}$, we use $\mathcal{P}_{\mathcal{X}}$ to denote the set of all possible distributions on $\mathcal{X}$. A vector of length $n$ is defined by $x^n=[x_1,\ldots, x_n]$. When all elements of the random vector $X^n$ are generated i.i.d according to distribution $P$, we denote it as $X^n\simiid P$. We use $S_n$, where $\vert S_n\vert =n!$, to denote the set of all possible permutations of a set of $n$ elements. For a permutation $\sigma\in S_n$, $\sigma_i$ denotes the $i$-th element of the permutation. $\lfloor x\rfloor_r$ is used to denote the remainder of $x$ divided by $r$. The indicator function of event $A$ is denoted by $\mathbbm{1}_{\{A\}}$. We use the notation $a_n\doteq e^{nb}$ when $\lim_{n\to \infty}\frac{\log a_n}{n}=b$.  
 
 $K_k\left(a_1,\ldots, a_{\binom{k}{2}}\right)$ is the complete graph with $k$ nodes with edge index $ i\in [1:\binom{k}{2}]$ and edge weights $a_i, \ i\in [1:\binom{k}{2}]$. We may drop the edge argument and simply write $K_k$ when the edge specification is not needed. 
 %  are not needed to be specified. 
 A cycle $c$ of length $r$ in $K_k$ may be interchangeably defined by a vector of vertices as $c^{(v)}=\left[v_1,\ldots, v_r\right]$ or by a set of edges $c^{(e)}=\left\{a_1,\ldots, a_r \right\}$ where $a_i$ is the edge between $(v_i, v_{i+1}),\forall i\in [1:r-1]$ and $a_r$ is that between $(v_r, v_1)$. With this notation, $c^{(v)}(i)$ is then used to indicate the $i$-th vertex of the cycle $c$.
  $C^{(r)}_k$ is used to denote the set of all cycles of length $r$ in the complete graph $K_k\left(a_1,\ldots, a_{\binom{k}{2}}\right)$. The cycle gain, denoted by $G(c)$, for cycle $c=\left\{a_1,\ldots, a_r \right\}\in C^{(r)}_k$ is the product of the edge weights within the cycle $c$, i.e., $G(c) = \prod_{i=1}^ra_i, \ \forall a_i\in c$. 
 %\nrd{Define a permutation $\sigma$ and what you mean by $\sigma(i)$.}
 
\section{Problem formulation}
\label{sec:problem}

Let $P:=\{P_1, \ldots , P_{A}\}\subset \mathcal{P}_{\mathcal{X}}$ consist of $A$ distinct distributions and also let $\Sigma$  be uniformly distributed over $S_A$, the set of permutations of $A$ elements. In addition, assume that we have $A$ independent random vectors $\{X_1^n,X_2^n,\ldots, X_{A}^n\}$ of length $n$ each.  
%For an unknown permutation $\sigma=\left( \sigma_1,\ldots,\sigma_k\right)\in S_k$, for each $i\in [1:k]$, the distribution  $P_{\sigma_i}^n$ is assigned to the random vector $X_i^n$.
For $\sigma$, a realization of $\Sigma$, assign the distribution $P_{\sigma_i}^n$ to the random vector $X_i^n, \forall i\in [1:A]$. 
After observing a sample $x^{nA}=[x_1^n,\ldots, x_A^n]$ of the random vector $X^{nA}=\left[X_1^n,\ldots, X_A^n \right]$, 
%where each $x_i^n$ is drawn independently according to the the underlying distribution  of $X_i^n$ (i.e., $P_{\sigma(i)}^n$), %the distribution of each random sequence $X_i^n, i\in [1:k]$. 
 we would like to identify $P_{\sigma_i}, \forall i\in[1:A]$.
 More specifically, we are interested in finding a permutation $\hsig: \mathcal{X}^{nA} \to S_A$ to indicate that $X_i^n\simiid P_{\hsig_i}, \ \forall i\in [1:A]$. Let $\Hsig = \hsig(X^{nA})$.%\nrd{why not $\hat{\Sigma}$?}

%In addition suppose that each  $x_i^n\in \{x_1^n,x_2^n,\ldots, x_{A_n}^n\}$ is generated independently and identically according to one and only one distinct distribution in $P$.    
%\nrd{Don't need a matching, want to characterize the number or matches that you can asymptotically distinguish or how large $A_n$ an be or? What is it you are after? This definition is not precise enough.}
%We would like to find a perfect matching between the sequences and distributions, i.e., a permutation $\hsig\left(x_1^n,\ldots, x_{A_n}^n \right)\in S_{A_n}$ to indicate a matching between $\left(x_i^n, P_{\hsig(i)}\right),  \forall i\in[1:A_n]$. 
The average probability of error for the set of distributions $P$ % and choosing $\hsig$ as the permutation of distributions of $[X_1^n,\ldots, X_k^n]$
 is given by %\nrd{why not $\PP\left[\hat{\sigma}\neq \sigma \right]$? }%\nrd{should $\hat{\sigma}$ then not also be an argument? why are you fixing $\hat{\sigma}$?}
\begin{align}
P_e^{(n)}&=\PP\left[\hat{\Sigma}\neq \Sigma \right]\notag\\
&=\frac{1}{(A)!}\sum_{\sigma \in S_{A}}\PP\left[\hat{\Sigma} \neq \sigma\vert X_i^n \simiid P_{\sigma_i}, \forall i\in [1:A] \right]\notag\\
&=\PP\left[\hat{\Sigma} \neq [1:A]\big\vert H_{(1,\ldots, A)} \right].~%\label{eq:symmetry}
\end{align}
%where~\eqref{eq:symmetry} is by symmetry of different hypothesis.
%\nrd{I think you mean to use "match" instead of "matching"}
where $H_{(1,\ldots, A)}:=\{X_i^n \simiid P_{i}, \forall i\in [1:A]\}$.

We say that a set of distributions $P$ are identifiable if $\lim_{n\to \infty}P_e^{(n)}\to 0$. 
%\nrd{above you define $P_e^{(n)}(P)$, so not precise enough, since $P_e^{n,k}$ not technically defined.}%\nrd{how is the $\sigma$ involved here? Is it still fixed?}
\begin{theorem}\label{thm:main}
A sequence of distributions $P=\{P_1,\ldots,P_{A_n}\}$ are identifiable iff
\[\lim_{n\to \infty}\sum_{\substack{1\leq i< j\leq A_n}}e^{-2nB(P_i,P_j)}=0,\]
where $B(P,Q)=B(Q,P)=-\log\PP_{P}\left[\left(\frac{Q}{P}\right)^{1/2} \right]$ is the Bhattacharya distance between  the distributions $P$ and $Q$.
\end{theorem}
%\nrd{Need the theorem statement.}
 
\begin{proof}
 As it is obvious from the result of Theorem~\ref{thm:main}, for the case that $A_n=A$ is a constant or the case that $A_n=O(n)$, the sequence of distributions in $P$ are always identifiable and the probability of error in the identification problem decays to zero as the blocklength $n$ goes to infinity. The interesting aspect of Theorem~\ref{thm:main} is in fact in the regime that $A_n$ increases exponentially with the blocklength. 

To prove Theorem~\ref{thm:main}, we provide  upper and lower bounds on the probability of error  in the following subsections.
\subsection{Upper bound on the probability of error}\label{subsec:achievability}
%\subsection{ML decoder}
%In order to minimize the probability of error, we  
We use the optimal  Maximum Likelihood (ML) decoder which minimizes the average probability of error, given by
\begin{align}
\hsig(x_1^n,\ldots, x_{A_n}^n):=\arg\max_{\sigma\in S_{A_n}}\sum_{i=1}^{A_n} \log \left(P_{\sigma_i}\left( x_i^n\right) \right),\label{eq:ml achieve}
\end{align}
where $P_{\sigma_i}\left( x_i^n\right)=\prod_{t=1}^nP_{\sigma_i}\left( x_{i,t}\right)$. 
The average probability of error associated with the ML decoder can also be written as
\begin{align}
&P_e^{(n)}=\PP\left[\Hsig \neq [1:A_n]\big\vert H_{(1,\ldots, A_n)} \right]\notag
\\
&= \PP\left[\bigcup_{\hsig\neq [1:A_n]}\Hsig=\hsig\big\vert H_{(1,\ldots, A_n)}\right]\notag
\\
%&= \PP\left[\bigcup_{i=1}^{A_n}\{\Hsig_i\neq i\}\big\vert X_i^n\simiid P_i,\forall i\in[1:A_n]\right]\notag
%\\
&=\PP\left[\bigcup_{r=2}^{A_n} \bigcup_{\substack{\hsig:\\\left\{\sum_{i=1}^{A_n}\mathbbm{1}_{\{\hsig_i\neq i\}}=r \right\}}} \!\!\!\!\Hsig =\hsig \big\vert H_{(1,\ldots, A_n)} \right]\label{eq:2<r}
\\
&=\PP\Bigg[\bigcup_{r=2}^{A_n} \bigcup_{\substack{\hsig:\\\left\{\sum_{i=1}^{A_n}\mathbbm{1}_{\{\hsig_i\neq i\}}=r \right\}}} \sum_{i=1}^{A_n}\log\frac{P_{\hsig_i}}{P_i}\left(X_i^n \right)\geq 0
%\notag\\
% &\qquad
  \big\vert H_{(1,\ldots, A_n)}\Bigg]
\label{eq:error double cycle}
\end{align}
where $\log \frac{P_i}{P_j}(X^n):=\sum_{t=1}^n\log \frac{P_i(X_t)}{P_j(X_t)}$ and where~\eqref{eq:2<r} is due to the requirement that each sequence is distributed according to a distinct distribution and hence the number of incorrect distributions ranges from $[2:A_n]$. Equation~\eqref{eq:error double cycle} is also the consequence of the ML decoder defined in~\eqref{eq:ml achieve}.
%\begin{align}
%P_e^{(n)}=\PP\left[\bigcup_{r=2}^{A_n} \left\{\sum_{i=1}^{A_n}\mathbbm{1}_{\hsig(i)\neq i}=r \right\}\big\vert x_i^n\sim P_i^n,\forall i\in[1:A_n] \right].\label{eq:error double cycle}
%\end{align}
In order to avoid considering the same set of error events multiple times, we incorporate a graph theoretic interpretation of $\left\{\sum_{i=1}^{A_n}\mathbbm{1}_{\{\Hsig_i\neq i\}}=r \right\}$ in~\eqref{eq:error double cycle}. Consider the two sequences $[i_1,\ldots, i_r]$ and $[\hsig_{i_1},\ldots, \hsig_{i_r}]$ for which we have 
\[\left\{\sum_{i=1}^{A_n}\mathbbm{1}_{\{\hsig_i\neq i\}}=\sum_{j=1}^{r}\mathbbm{1}_{\{\hsig_{i_j}\neq i_j\}}=r \right\}.\]
These two sequences in~\eqref{eq:error double cycle} in fact indicate the event that we have (incorrectly) identified $X_{i_j}^n\simiid P_{\hsig_{i_j}}$ instead of the (true) distribution  $X_{i_j}^n\simiid P_{{i_j}}, \forall j\in [1:r]$. 
 For a complete graph $K_{A_n}$, the set of edges between $\left((i_1,\hsig_{i_1}),\ldots, (i_r,\hsig_{i_r})\right)$ in $K_{A_n}$ would produce a single cycle of length $r$ or a set of disjoint cycles with total length $r$. However, we should note that in the latter case where the sequence of edges construct a set of (lets say of size $L$) disjoint cycles (each with some length $\tilde{r}_l$ for $\tilde{r}_l<r$ such that $\sum_{l=1}^L \tilde{r}_l=r$), then those cycles and their corresponding  sequences are already taken into account in the (union of) set of $\tilde{r}_l$ error events. 

As an example, assume $A_n=4$ and consider the error event
\[\log\frac{P_2}{P_1}(X_1^n)+\log\frac{P_1}{P_2}(X_2^n)+\log\frac{P_4}{P_3}(X_3^n)+\log\frac{P_3}{P_4}(X_4^n)\geq 0,\]
which corresponds to the (error) event of choosing $[\hsig_1, \hsig_2,\hsig_3,\hsig_4]=[2,1, 4, 3]$ over $[1,2,3,4]$ with $r=4$ errors. In the graph representation, this gives two cycles of length $2$ each, which correspond to
\begin{align*}
\log\frac{P_2}{P_1}(X_1^n)+\log\frac{P_1}{P_2}(X_2^n)\geq 0\ \cap\\
\log\frac{P_4}{P_3}(X_3^n)+\log\frac{P_3}{P_4}(X_4^n)\geq 0,
\end{align*} 
and are already accounted for in the events $\left\{[\hsig_1, \hsig_2,\hsig_3,\hsig_4]=[2,1, 3, 4]\right\}\cup \left\{[\hsig_1, \hsig_2,\hsig_3,\hsig_4]=[1,2, 4, 3]\right\}$ with $r=2$.
 
 As the result, in order to avoid double counting, in calculating the value of~\eqref{eq:error double cycle} for each $r$ we should only consider the sets of sequences which produce a {\it single} cycle of length $r$. Hence, we can upper bound the probability of error in~\eqref{eq:error double cycle} as (where we drop the conditioning for ease of notation)
%\begin{align}
%P_e^{(n)}&\leq \sum_{r=2}^{A_n} 
%\sum_{\substack{\{i_1,\ldots i_r\}\\\subseteq [1:A_n]}} \sum_{\substack{\hsig \in S_r\\ i_{\hsig(j)}\neq i_j,\forall j\in[1:r]}} \PP\left[ x_{i_j}^n\sim P_{i_{\hsig(j)}}^n, \forall j\in [1:r]\right]\notag
%\\&\leq \sum_{r=2}^{A_n} 
%\sum_{\substack{\{i_1,\ldots i_r\}\\\subseteq [1:A_n]}} \sum_{\substack{\hsig \in S_r\\i_{\hsig(j)}\neq i_j,\forall j\in[1:r]}} e^{-n\sum_{j=1}^r B(P_{i_j}, P_{\hsig(j)})}\label{eq:ml achievability}
%\end{align}
\begin{align}
P_e^{(n)}&\leq\sum_{r=2}^{A_n} 
\sum_{\substack{c\in C^{(r)}_{A_n}}} \PP\left[\sum_{i=1}^r\log \frac{P_{\lfloor c^{(v)}(i+1)\rfloor_{r}}}{P_{c^{(v)}(i)}}\left(X^n_{c^{(v)}(i)}\right)\geq 0\right]\notag\\
%P_e^{(n)}&\leq \sum_{r=2}^{A_n} 
%\sum_{\substack{c\in C^{(r)}_{A_n}}} \PP\left[ X_{c^{(v)}(i)}^n\simiid P_{ c^{(v)}(\lfloor i+1\rfloor_{r})}, \forall i\in [1:r]\right]\notag
%\\
&\leq \sum_{r=2}^{A_n} 
\sum_{\substack{c\in C^{(r)}_{A_n}}} e^{-n\sum_{i=1}^r B(P_{c^{(v)}(i)}, P_{c^{(v)}(\lfloor i+1\rfloor_{r})})}\label{eq:ml achievability}
\\
&=\sum_{r=2}^{A_n} \sum_{c\in C_{A_n}^{(r)}} G(c)\label{eq:achieve graph},
\end{align}
where $r$ enumerates the number of incorrect matchings %and $i_1,\ldots, i_r$ is the indices of the distributions in error. 
and where $c(i)$ is the $i$-th vertex in the cycle $c$. The inequality in~\eqref{eq:ml achievability} is by
\begin{align}
% &\PP\left[ X_{c^{(v)}(i)}^n\simiid P_{\lfloor c^{(v)}(i+1)\rfloor_{r}}, \forall i\in [1:r]\right]\notag \\
 &
 \PP\left[\sum_{i=1}^r\log \frac{P_{\lfloor c^{(v)}(i+1)\rfloor_{r}}}{P_{c^{(v)}(i)}}\left(X^n_{c^{(v)}(i)}\right)\geq 0\right]\notag
\\&\leq \exp\left\{n\inf_t \log \mathbb{E}\left[ \prod_{i=1}^r \left(\frac{P_{ c^{(v)}(\lfloor i+1\rfloor_{r})}}{P_{c^{(v)}(i)}}\left(X^n_{c(i)}\right)\right)^t \right]\right\}\notag
\\&\leq  \exp\left\{n\sum_{i=1}^r \log \mathbb{E}\left[  \left(\frac{P_{c^{(v)}(\lfloor i+1\rfloor_{r})}}{P_{c^{(v)}(i)}}\left(X^n_{c(i)}\right)\right)^{1/2} \right]\right\}\label{eq:t = 1/2}
\\&=  \exp\left\{-n\sum_{i=1}^r B(P_{c^{(v)}(i)}, P_{ c^{(v)}(\lfloor i+1\rfloor_{r})})\right\}.\notag
\end{align}
In~\eqref{eq:achieve graph}, we have also defined $e^{-nB(P_i,P_j)}$ to be the edge weight between vertices $(i,j)$  in the complete graph $K_{A_n}$. Hence $ G(c)=e^{-n\sum_{i=1}^r B(P_{c^{(v)}(i)}, P_{c^{(v)}(\lfloor i+1\rfloor_{r})})}$ is the  gain of cycle $c$. 

%It is true that~\eqref{eq:ml achievability} seems very complicated but it should be noted that $|P^r|\leq r!$ and $|i_1,\ldots, i_r |=\binom{A_n}{r}$ and if hypothetically all $B(P_i, P_j)=B$, then~\eqref{eq:ml achievability} can be upper bounded as 
%\begin{align}
%\PP[\text{error}]\leq \sum_{r=2}^{A_n} \binom{A_n}{r} r! e^{rnB}\leq \sum_{r=2}^{A_n} A_n^r e^{rnB} \approx A_n^2 e^{2nB},
%\end{align}
%\textcolor{red}{
%\begin{align}
%\PP[\text{error}]\leq \sum_{r=2}^{A_n} \binom{A_n}{r} (r-1)! e^{rnB}
%\end{align}
%}
%which says $\nu<B$. This is already better than~\eqref{eq:seq ml order}.
The fact that we used $t=1/2$ in~\eqref{eq:t = 1/2} instead of finding the exact optimizing $t$, comes from  the fact that $t=1/2$ is the optimal choice for $r=2$ and as we will see later, the rest of the error events are dominated by the set $2$ incorrect distributions. This can be seen as follows for $X_1^n\simiid P_1, X_2^n\simiid P_2$
\begin{align}
&\PP\left[\log\frac{P_1}{P_2}(X^n_2)+\log\frac{P_2}{P_1}(X_1^n)\geq 0 \right]\notag
\\
&=\sum_{\substack{\hat{P}_1,\hat{P_2}:\\\sum_{x\in \mathcal{X}}\hat{P}_1(x)\log \frac{P_2(x)}{P_1(x)}+\\\hat{P_2}(y)\log\frac{P_1(x)}{P_2(x)}\geq 0}}\exp\left\{-nD\left(\hat{P}_1\parallel P_1\right)-nD\left(\hat{P}_2\parallel P_2 \right)\right\}\notag
\\&\doteq e^{-nD\left(\tilde{P} \parallel P_1\right)-nD\left(\tilde{P} \parallel P_2 \right)}=e^{-2nB(P_1, P_2)},\label{eq:lagrange}
\end{align}
where $\tilde{P}$ in the first equality in~\eqref{eq:lagrange}, by using the Lagrangian method, can be shown to be equal to $\tilde{P}(x)=\frac{\sqrt{P_1(x)P_2(x)}}{\sum_{x'}\sqrt{P_1(x')P_2(x')}}$ and subsequently the second inequality in~\eqref{eq:lagrange} is proved.

In order to further simplify the expression in~\eqref{eq:achieve graph}, we use the following graph theoretic Lemma, the proof of which is given in the Appendix.

\begin{lemma}\label{lemma:graph}
In a complete graph $K_{k}\left(a_1,\ldots, a_{n_k}\right)$ and for the set of cycles of length $r, \mathcal{C}_{k}^{(r)}=\{c_1,\ldots c_{N_{r,k}}\}$ we have
\begin{align}
\frac{1}{N_{r,k}}\left( G(c_1)+\ldots G(c_{N_{r,k}})\right)& \leq \left(\frac{a_1^2+\ldots+ a_{n_k}^2}{n_k}\right)^{\frac{r}{2}}\notag
%\\&= M_{\frac{2}{r}}(a_1^r,\ldots, a_n^r),\notag
\end{align}
where $N_{r,k}, n_k$ are the number of cycles of length $r$ and the number of edges in the complete graph $K_k$, respectively.
\end{lemma}
By Lemma~\ref{lemma:graph} and~\eqref{eq:achieve graph} we can write
\begin{align}
P_e^{(n)}&\leq \sum_{r=2}^{A_n}\sum_{c\in C_{A_n}^{(r)}}G(c)\notag\\
&\leq \sum_{r=2}^{A_n} \frac{N_{r,A_n}}{\left({n_{A_n}}\right)^{\frac{r}{2}}}\left( a_1^2+\ldots+a_{n_{A_n}}^2\right)^{r/2}\notag\\
&\leq \sum_{r=2}^{A_n} 4^r\left( \sum_{\substack{1\leq i<j\leq A_n}}e^{-2nB(P_i,P_j)}\right)^{r/2}\label{eq:nN}\\
&\leq \frac{16\left(\sum_{\substack{1\leq i<j\leq A_n}}e^{-2nB(P_i,P_j)}\right)}{1-4\sqrt{\sum_{\substack{1\leq i<j\leq A_n}}e^{-2nB(P_i,P_j)}}}\label{eq:achieve final ub},
\end{align} 
where~\eqref{eq:nN} is by Fact 1 (see Appendix) and
\[\frac{N_{r,A_n}}{\left( {n_{A_n}}\right)^{r/2}}=\frac{\binom{A_n}{r}(r-1)! /2}{\left(\binom{A_n}{2}\right)^{r/2}}\leq 4^{r}.\]
 The upper bound on the  probability of error in~\eqref{eq:achieve final ub}  goes to zero if 
\[\lim_{n\to \infty}\sum_{\substack{1\leq i<j\leq A_n}}e^{-2nB(P_i,P_j)}=0.\]
As a result of Lemma~\ref{lemma:graph}, it can be seen from~\eqref{eq:nN} that the sum of probabilities that $r\geq 3$ distributions are incorrectly identified is dominated by the probability that only $r=2$ distributions are incorrectly identified. This shows that the most probable error event is indeed the error events with two wrong distributions.
%\subsection{Sub-optimal sequential decoder}
%If we perform sequential (ML) decoding for each block, we will end up with a suboptimal result as we have not incorporated the fact that each distribution happens once and only once. To see this, we choose the distribution with the largest likelihood for each sequence. As the result, we end up with the following error bound: \nrd{need to explain more, if we decide to put it here. The second problem is also that in this procedure,  I am also upper bounding the probability of error.. so I can't really say this is definitely worst than the optimal test... right?} \nbl{I agree, since it's an upper bound you are just comparing two upper bounds. You can still put a remark that a "naive application of a sequential test" seems to be sub-optimal because... but make no formal claims?}
%\begin{align}
%\PP[\text{error}]&\leq \sum_{i=1}^{A_n} \sum_{\substack{j=1\\j\neq i}}^{A_n}\PP\left[\log \frac{P_i}{P_j}(X_i^n)<0\right]\notag\\
%&
%\leq \sum_{\substack{i,j\in [1:A_n]^2\\j\neq i}} e^{-nB\left(P_i,P_j \right)}.\label{eq:seq ml}
%\end{align}
%As it can be seen from~\eqref{eq:seq ml} and~\eqref{eq:achieve final ub}, discarding the assumption that each sequence is generated by a distinct distribution will lead to a larger bound on the probability of error \nbl{though this rough argument does not rigorously show the sub-optimality of ML per-se.}
%error bound.

%-----------------------
\subsection{Lower bound on the probability of error}
For our converse, we use the optimal ML decoder, and as a lower bound to the probability of error in~\eqref{eq:error double cycle}, we only consider the set of error events with only two incorrect distributions, i.e. the set of events with $r=2$. In this case we have
\begin{align}
P_e^{(n)}&\geq \PP\left[ \bigcup_{\substack{1\leq i<j\leq A_n}} \log\frac{P_i}{P_j}(X^n_j)+\log\frac{P_j}{P_i}(X_i^n)\geq 0\right]\notag
\\&\geq \frac{\left(\sum_{\substack{1\leq i<j\leq A_n}}  \PP\left[\xi_{i,j}\right]\right)^2}{\sum_{\substack{(i,j),(j,k)\\(i,j)\neq(l,k)\\i\neq j, l\neq k}}\PP[\xi_{i,j},\xi_{k,l}]}
%\\&\geq
%\frac{\left(\mathbb{E}\left[\sum_{\substack{i,j\\i\neq j}}\xi_{i,j} \right] \right)^2}{\mathbb{E}\left[\left(\sum_{\substack{i,j\\i\neq j}}\xi_{i,j} \right)^2 \right]}
\label{eq:conv lb}
\end{align}
where~\eqref{eq:conv lb} is by~\cite{chung1952application}
%$\PP[X>0]\geq \frac{\left(\mathbb{E}[X]\right)^2}{\mathbb{E}\left[X^2\right]},$
 and where 
%\[\PP[\xi_{i,j}=1]=\PP\left[  \log\frac{P_i}{P_j}(X^n_j)+\log\frac{P_j}{P_i}(X_i^n)\geq 0\right] \]
%and
\begin{align}
&\xi_{i,j}:= \left\{\log\frac{P_i}{P_j}(X^n_j)+\log\frac{P_j}{P_i}(X_i^n)\geq 0\right\}.
\label{eq:1/2 optimal}
\end{align}
%and where~\eqref{eq:1/2 optimal} is by~\eqref{eq:lagrange}.
We upper bound the denominator of~\eqref{eq:conv lb} by  
\begin{align}
&\!\PP[\xi_{i,j}, \xi_{i,k}]=\PP\left[ \log\frac{P_i}{P_j}(X^n_j)+\log\frac{P_j}{P_i}(X_i^n)\geq 0\ \right.\notag
\\&
\qquad  \left.\cap\ \log\frac{P_i}{P_k}(X^n_k)+\log\frac{P_k}{P_i}(X_i^n)\geq 0 \right]\notag
\\&
\leq \PP \left[ \log\frac{P_i}{P_j}(X^n_j)+\log\frac{P_j}{P_i}(X_i^n) \right.\notag
\\
&\qquad \left.  +\log\frac{P_i}{P_k}(X^n_k)+\log\frac{P_k}{P_i}(X_i^n) \geq 0 \right]\notag
\\&
\leq \exp\Bigg\{n \inf_t  \notag\\
&\log\left( \mathbb{E}\left[ \left(\!\frac{P_i}{P_j}(X^n_j) \cdot \frac{P_j}{P_i}(X_i^n) \cdot \frac{P_i}{P_k}(X^n_k)\cdot \frac{P_k}{P_i}(X_i^n)\right)^t \right]\right) \Bigg\}\notag
\\&
\!\!\leq \exp\left\{\!n \log \mathbb{E}\!\!\left[\! \left(\!\frac{P_i}{P_j}(X^n_j)\!\cdot\!\frac{P_j}{P_i}(X_i^n)\!\cdot \! \frac{P_i}{P_k}(X^n_k)\! \cdot\! \frac{P_k}{P_i}(X_i^n)\!\right)^{\frac{1}{2}} \!\right] \!\right\}\notag
\\&= \exp\left\{-n B({P_i},{P_j})-nB({P_j},P_k)-nB({P_i},{P_k}) \right\}.\label{eq:covariance}
\end{align}
An upper bound for $\PP\left[\xi_{i,j},\xi_{k,l} \right]$ can be derived accordingly.
By substituting~\eqref{eq:lagrange} and~\eqref{eq:covariance} in~\eqref{eq:conv lb} we have
%\nrd{EDAS may not accept lines that spill over into the margin (had this with another paper) -- try it out right now to see if it will be accepted and change it if not.}
\begin{align}
&P_e^{(n)}\geq \notag\\
&\hspace{-.5cm}\frac{\left(\sum_{1\leq i<j\leq A_n}e^{-2nB(P_i,P_j)}\right)^2}{\sum_{i,j,k}e^{-nB(P_i,P_j)-nB(P_i,P_k)-nB(P_k,P_j)}\!+\!\left(\sum_{i,j}e^{-2nB(P_i,P_j)}\!\right)^2}\notag
\\&\hspace{-.5cm}\geq \frac{\left(\sum_{i,j}e^{-2nB(P_i,P_j)}\right)^2}{8\left(\sum\limits_{1\leq i<j\leq A_n}e^{-2nB(P_i,P_j)}\right)^{3/2}\!\!\!+\left(\sum\limits_{1\leq i<j\leq A_n}e^{-2nB(P_i,P_j)}\right)^2}
\label{eq:conv simplified}\\
&=\frac{\sqrt{\sum_{1\leq i<j\leq A_n}e^{-2nB(P_i,P_j)}}}{8+\sqrt{\sum_{1\leq i<j\leq A_n}e^{-2nB(P_i,P_j)}}}, \label{eq:conv final lb}
\end{align}
where~\eqref{eq:conv simplified} is by Lemma~\ref{lemma:graph}. %In order to have a \nrd{? zero lower bound} on the probability of error in~\eqref{eq:conv final lb}, we need to make sure that 
As it can be seen from~\eqref{eq:conv final lb}, if  $\lim_{n\to \infty }\sum_{\substack{1\leq i<j\leq A_n}}e^{-2nB(P_i,P_j)}\neq 0$, the probability of error is bounded away from zero. As the result, we have to have 
$\lim_{n\to \infty }\sum_{\substack{1\leq i<j\leq A_n}}e^{-2nB(P_i,P_j)}=0$, which also matches our upper bound on probability of error  in~\eqref{eq:achieve final ub}.
%Again,~\eqref{eq:conv hard} is hard to compute, but if we assume hypothetically that $B(P_i,P_j)=B$, we get
%\begin{align}
% \PP[\text{error}]&\geq 1-\frac{A_n^2e^{2nB}+A_n^3e^{3nB}}{(A_n^2e^{2nB})^2}\notag
%\\&=1-\frac{1+A_ne^{nB}}{A_n^2e^{2nB}}
%\end{align}
%which says that if $\nu>B$ then probability of error is bounded away from zero. This order is similar to the bounded we got in our achievability with using ML decoder.
 \end{proof}

\section{Conclusion}
In this paper, we generalized the identification problem to the case that the number of distributions grows with the blocklength $n$. We found matching upper and lower bounds on the probability of identification error. This result characterizes the relation between the number of distributions and the pairwise distance of the distributions for reliable identification.

%\section{ACKNOWLEDGMENT}
%The work of the authors was partially funded by NSF under award 1422511. The contents of this article are solely the responsibility of the authors and do not necessarily represent the official views of the NSF.
\appendix
We first consider the case that r is an even number and then prove
%We first even $r$ and prove 
 \begin{align}
r\!\left(n_k\right)^{\frac{r}{2}-1}\left( G(c_1)+\ldots G(c_{N_{r,k}})\right) \!\leq\! \frac{N_{r,k}r}{n_k}\! \left(a_1^2+\ldots+ {a_{n_k}}^2\right)^{\frac{r}{2}}\!.\label{eq:main ineq}
\end{align}
We may drop the subscripts and use $N:=N_{r,k}$ and $n:=n_k$ in the following for notational ease.
Our goal is to expand the right hand side (RHS) of~\eqref{eq:main ineq} such that all elements have coefficient $1$. Then, we parse these elements into $N$ different groups (details will be provided later) such that using the AM-GM inequality (i.e., $ n\left(\prod_{i=1}^n a_i\right)^{\frac{1}{n}}\leq \sum_{i=1}^n a_i$) on each group, we get one of the $N$ terms on the LHS of~\eqref{eq:main ineq}. Before stating the rigorous proof, we provide an example of this strategy for the graph with $k=4$ vertices shown in Fig.~\ref{fig:G6}. In this example, we consider the Lemma for $r=4$ cycles (for which we have $N=3$). 
\begin{figure}[htbp]
\centering
\includegraphics[width=.2\textwidth]{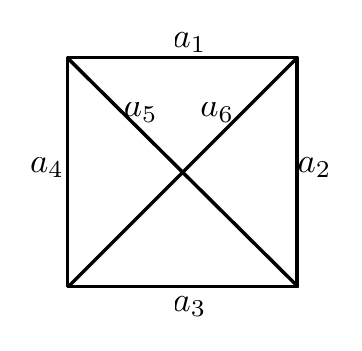}
\caption{A complete graph with $4$ vertices}
\label{fig:G6}
\end{figure}

%For the graph in Fig.~\ref{fig:G6}, 
We may expand the RHS in~\eqref{eq:main ineq} as
\begin{align*}
%&24\left(a_1a_2a_3a_4+a_1a_6a_3a_5+a_4a_5a_2a_6\right) \notag
&\qquad2\left(a_1^2+\ldots +a_6^2 \right)^2
=\Theta_1+\Theta_2+\Theta_3,\\
&\Theta_1\!=\!\big\{a_1^4+a_2^4+a_3^4+a_4^4+a_1^2a_3^2+a_1^2a_3^2+a_2^2a_4^2+a_2^2a_4^2\\
&+a_1^2a_2^2+a_1^2a_2^2+a_1^2a_2^2+a_1^2a_2^2+a_1^2a_4^2+a_1^2a_4^2+a_1^2a_4^2+a_1^2a_4^2
\\&+a_2^2a_3^2+a_2^2a_3^2+a_2^2a_3^2+a_2^2a_3^2+a_3^2a_4^2+a_3^2a_4^2+a_3^2a_4^2+a_3^2a_4^2\big\}
\\&\Theta_2\!=\!\big\{a_1^4+a_6^4+a_3^4+a_5^4+a_5^2a_6^2+a_5^2a_6^2+a_1^2a_3^2+a_1^2a_3^2
\\&
+a_1^2a_6^2+a_1^2a_6^2+a_1^2a_6^2+a_1^2a_6^2+a_1^2a_5^2+a_1^2a_5^2+a_1^2a_5^2+a_1^2a_5^2
\\&+a_3^2a_6^2+a_3^2a_6^2+a_3^2a_6^2+a_3^2a_6^2+a_3^2a_5^2+a_3^2a_5^2+a_3^2a_5^2+a_3^2a_5^2\big\}
\\&\Theta_3\!=\!\big\{a_4^4+a_5^4+a_2^4+a_6^4+a_5^2a_6^2+a_5^2a_6^2+a_2^2a_4^2+a_2^2a_4^2
\\&
+a_4^2a_5^2+a_4^2a_5^2+a_4^2a_5^2+a_4^2a_5^2
+a_4^2a_6^2+a_4^2a_6^2+a_4^2a_6^2+a_4^2a_6^2
\\&
+a_2^2a_5^2+a_2^2a_5^2+a_2^2a_5^2+a_2^2a_5^2+a_2^2a_6^2+a_2^2a_6^2+a_2^2a_6^2+a_2^2a_6^2\big\}.
\end{align*}
It can be easily seen that if we use the AM-GM inequality on $\Theta_1$, $\Theta_2$ and $\Theta_3$, we can get the lower bound equal to $24(a_1a_2a_3a_4), 24(a_1a_6a_3a_5)$ and $24(a_4a_5a_2a_6)$, respectively where $rn^{\frac{r}{2}-1}=24$ and hence~\eqref{eq:main ineq} holds in this example.

We proceed to prove  Lemma~\ref{lemma:graph} for arbitrary $k$ and (even) $r\geq 2$. We propose the following scheme   to group the elements on the RHS of~\eqref{eq:main ineq} and then we prove that this grouping indeed leads to the claimed inequality in the Lemma. 

{\bf Grouping scheme:} For each cycle $c_i =\{a_{i_1}\ldots , a_{i_r}\}$, we need a group of elements, $\Theta_i$, from the RHS of~\eqref{eq:main ineq}. In this regard, we consider all possible subsets of the edges of cycle $c_i$ with $1:\frac{r}{2}$ elements (e.g. $\left\{\{a_{i_1}\},\ldots \{a_{i_1},a_{i_2}\},\ldots \{a_{i_1}\ldots, a_{i_{r/2}}\},\ldots\right\} $). For each one of these subsets, we find the respective elements from the RHS of~\eqref{eq:main ineq} that is the multiplication of the elements in that subset. For example, for the subset $\{a_{i_1},a_{i_2},a_{i_3}\}$, we consider the elements like $a_{i_1}^{n_{i_1}}a_{i_2}^{n_{i_2}}a_{i_3}^{n_{i_3}}$ for all possible $n_{i_1},n_{i_2},n_{i_3}>0$ from the RHS of~\eqref{eq:main ineq}. However, note that we do not assign all such elements to cycle $c_i$ only. If there are $l$ cycles of length $r$ that all contain $\{a_{i_1},a_{i_2},a_{i_3}\}$, we should assign $\frac{1}{l}$ of the elements like $a_{i_1}^{n_{i_1}}a_{i_2}^{n_{i_2}}a_{i_3}^{n_{i_3}}, \ n_{i_1},n_{i_2},n_{i_3}>0$ to cycle $c_i$ (so that we can assign the same amount of elements to other cycles with similar edges). 

We state some facts, which  can be easily verified:

{\bf Fact 1.} In a complete graph $K_k$, there are $N=N_{r,k}=\binom{k}{r}\frac{(r-1)!}{2}$ cycles of length $r$. 

{\bf Fact 2.} By expanding the RHS of~\eqref{eq:main ineq} such that all elements have coefficient $1$, we end up with $\left(\frac{N r}{n}\right) n^{\frac{r}{2}}$ elements.

{\bf Fact 3.} Expanding the RHS of~\eqref{eq:main ineq} such that all elements have coefficient $1$, and finding their product yields
\[\left(a_1\times\ldots\times a_n \right)^{\left(\frac{Nr}{n}\right)rn^{ \frac{r}{2}-1}}.\]

{\bf Fact 4.} In above grouping scheme each element on the RHS of~\eqref{eq:main ineq} is summed in exactly one group. Hence, by symmetry and Fact 2, each group is the sum of $r n^{\frac{r}{2}-1}$ elements.
%\begin{remark}
%It should be noted that in~\eqref{eq:main ineq}, we have multiplied both sides by the factor $\frac{N_{r,k}r}{n}$, which is the number of cycles of length $r$ in $K_{A_n}$ who contain $a_i$ in their gains. Since in all these cycles we need to add 
%\end{remark}

Now, consider any two cycles $c^{(e)}_i=\{a_{i_1},\ldots , a_{i_r}\},c^{(e)}_j=\{a_{j_1},\ldots , a_{j_r}\} $. Assume that using the above grouping scheme, we get the group of elements $\Theta_i,\Theta_j$ (where by fact 3 each one is the sum of $r n^{\frac{r}{2}-1}$ elements).
If we apply the AM-GM inequality on each one of the two groups, we get 
\begin{align*}
\Theta_i\geq r n^{\frac{r}{2}-1} \left( a_{i_1}^{n_{i_1}}\times \ldots\times a_{i_r}^{n_{1_r}}  \right)^{\left(\frac{1}{r n^{\frac{r}{2}-1}}\right)},  \\
\Theta_j\geq r n^{\frac{r}{2}-1} \left( a_{j_1}^{n_{j_1}}\times \ldots\times a_{j_r}^{n_{j_r}}  \right)^{\left(\frac{1}{r n^{\frac{r}{2}-1}}\right)},
\end{align*}
where $\prod_{t=1}^r a_{i_t}^{n_{i_t}}$ is the product of the elements in $\Theta_i$.
 By symmetry of the grouping scheme for different cycles, it is obvious that $\forall t\in[1:r], n_{i_t}=n_{j_t}$. Hence $ n_{i_t}=n_{j_t}=p_t,\forall i,j\in [1:N]$. %i.e., using the color to indicate equalities, we have \nrd{I do not understand the use of color, and this may be a little strange for people who print --try a different way to get this across?}
 i.e., we have
 \begin{align}
  \Theta_i &\geq r n^{\frac{r}{2}-1} \left( a_{i_1}^{p_1}\times \ldots\times a_{i_r}^{p_r}  \right)^{\left(\frac{1}{r n^{\frac{r}{2}-1}}\right)}\label{eq:p}.
% \Theta_1 &\geq r n^{\frac{r}{2}-1} \left( a_{1_1}^{p_1}\times \ldots\times a_{1_r}^{p_r}  \right)^{\left(\frac{1}{r n^{\frac{r}{2}-1}}\right)}
% \\
% &\vdots\\
%\Theta_{N}&\geq r n^{\frac{r}{2}-1} \left( a_{{N}_1}^{p_1}\times \ldots\times a_{{N}_r}^{p_r}  \right)^{\left(\frac{1}{r n^{\frac{r}{2}-1}}\right)}.
 \end{align}

 By symmetry of the grouping scheme over the elements of each cycle, we also get that $n_{i_k}=n_{i_l}=q_i,\forall k,l\in [1:r]$. i.e.
 \begin{align}
 \Theta_i \geq r n^{\frac{r}{2}-1} \left( a_{i_1}^{q_i}\times \ldots\times a_{i_r}^{q_i}  \right)^{\left(\frac{1}{r n^{\frac{r}{2}-1}}\right)}.\label{eq:q}
 \end{align}
 It can be seen from~\eqref{eq:p} and~\eqref{eq:q} that all the elements of all groups have the same power  $n_{i_t}=p,\forall i\in[1:N], t\in [1:r]$. i.e.,
%  \begin{align}
% \Theta_1 &\geq r n^{\frac{r}{2}-1} \left( a_{1_1}^{p}\times \ldots\times a_{1_r}^{p}  \right)^{\left(\frac{1}{r n^{\frac{r}{2}-1}}\right)}\notag
% \\
% &\vdots\notag\\
%\Theta_{N}&\geq r n^{\frac{r}{2}-1} \left( a_{{N}_1}^{p}\times \ldots\times a_{{N}_r}^{p}  \right)^{\left(\frac{1}{r n^{\frac{r}{2}-1}}\right)}.\label{eq:conclusion}
% \end{align}
  \begin{align*}
 \Theta_i &\geq r n^{\frac{r}{2}-1} \left( a_{i_1}^{p}\times \ldots\times a_{i_r}^{p}  \right)^{\left(\frac{1}{r n^{\frac{r}{2}-1}}\right)}.
 \end{align*}
 Since each element on the RHS of~\eqref{eq:main ineq} is assigned to one and only one group and since $\prod_{t=1}^r a_{i_t}^{n_{i_t}}= \prod_{t=1}^r a_{i_t}^p$ is the product of the elements of each group $\Theta_i$, the product of all elements in $\Theta_1+\ldots +\Theta_{N}$ (which is equal to product of the elements in the expanded version of the RHS of~\eqref{eq:main ineq}) is
 $\prod_{i=1}^{N}\prod_{t=1}^r a_{i_t}^p$.
 
 In addition, since each $a_i$ appears in exactly $\frac{Nr}{n}$ of the cycles, by Fact 3 and a double counting argument, we have 
 \[p\times \frac{Nr}{n}=\left( \frac{Nr}{n}\right) rn^{ \frac{r}{2}-1},\]
 and hence $p=rn^{ \frac{r}{2}-1}$.
 Hence, the lower bound of the AM-GM inequality on the $\Theta_1+\ldots+ \Theta_{N}$, will result in
  \[rn^{ \frac{r}{2}-1} G(c_1)+\ldots+ rn^{ \frac{r}{2}-1} G(c_{N_r}),\] and  the Lemma is proved for even $r$.
 
 For odd values of $r$, the problem that may arise by using the grouping strategy in its current form, is when $r<\frac{k}{2}$. In this case, some of the terms on the RHS of~\eqref{eq:main ineq} may contain   multiplication of $a_i$'s that are not present in any of the $G(c_i)$'s. To overcome this, take both sides to the power of $2m$ for the smallest $m$ such that $rm>\frac{k}{2}$. Then the RHS of~\eqref{eq:main ineq} is at most the multiplication of $rm$ different $a_i$'s and on the LHS of~\eqref{eq:main ineq},  there are $2m$ cycles of length $r$ multiplied together. By our choice of $2m$, now, all possible combinations of $a_i$'s on the RHS are present in at least one cycle multiplication in the LHS. Hence, we can now continue the proof with the same strategy as even values of $r$ for the odd values of $r$.

\bibliography{refs}
\bibliographystyle{IEEEtran}

\end{document}